\documentclass{llncs}
\usepackage{color}
\usepackage{graphicx}
\usepackage{amsfonts}
\usepackage{latexsym,amssymb,amsmath}
\usepackage{fullpage}

\usepackage[ruled]{algorithm}
\usepackage[noend]{algpseudocode}
\algdef{SE}[DOWHILE]{Do}{doWhile}{\algorithmicdo}[1]{\algorithmicwhile\ #1}
\usepackage{comment}
\usepackage{epstopdf}
\usepackage{soul}

\usepackage{caption}
\usepackage{subcaption}
\usepackage{wrapfig}
\usepackage[hidelinks]{hyperref}

\usepackage[appendix=inline]{apxproof}

\newtheoremrep{lemma}{Lemma}[section]{\bfseries}{\itshape}
\newtheoremrep{theorem}{Theorem}[section]{\bfseries}{\itshape}

\begin{document}

\title{Message Delivery in the Plane by Robots with Different Speeds
}


\author{
    Jared Coleman\inst{1}
    \and
    Evangelos Kranakis\inst{2}\inst{5}
    \and
    Danny Krizanc\inst{3}
    \and
    Oscar Morales Ponce\inst{4}
}

\institute{
    Department of Computer Science, University of Southern California, 
    Los Angeles, CA, USA
    \and
    School of Computer Science, Carleton University, Ottawa, Ontario, Canada.
    \and
    Department of Mathematics \& Computer Science, Wesleyan University, 
    Middletown CT, USA.
    \and
    Department of Computer Science, California State University, Long Beach
    \and
    Research supported in part by NSERC Discovery grant.
}

\maketitle

\begin{abstract}
We study a fundamental cooperative message-delivery problem on the plane.
Assume $n$ robots which can move in any direction, are placed arbitrarily on the 
plane.
Robots each have their own maximum speed and can communicate with each other 
face-to-face (i.e., when they are at the same location at the same time). 
There are also two designated points on the plane, $S$ (the {\em source}) and 
$D$ (the {\em destination}). 
The robots are required to transmit the message from the source to the 
destination as quickly as possible by face-to-face message passing.
We consider both the offline setting where all information (the locations and 
maximum speeds of the robots) are known in advance and the online setting where 
each robot knows only its own position and speed along with the positions of 
$S$ and $D$.

In the offline case, we discover an important connection between the problem 
for two-robot systems and the well-known Apollonius circle which we employ to 
design an optimal algorithm. 
We also propose a $\sqrt 2$ approximation algorithm for systems with any number 
of robots.
In the online setting, we provide an algorithm with competitive ratio 
$\frac 17 \left( 5+ 4 \sqrt{2} \right)$ for two-robot systems and show that the 
same algorithm has a competitive ratio less than $2$ for systems with any number 
of robots.
We also show these results are tight for the given algorithm. 
Finally, we give two lower bounds (employing different arguments) on the 
competitive ratio of any online algorithm, one of $1.0391$ and the other of 
$1.0405$.

\vspace{0.1cm}
\end{abstract}

\section{Introduction}\label{sec:intro}
We study the problem of delivering a message in minimum time from a source to a 
destination using autonomous mobile robots with different maximum speeds.
We extend the work on this communication problem studied previously on graphs~\cite{AnayaCCLPV16,Bartschi0M18,carvalho2019efficient,carvalho2021fast}. 
In our setting, the robots are initially distributed in arbitrary locations in 
the plane and the locations of the source and destination are known by all. 
The robots may move with their own (maximum) speed. 
Robots cooperate by exchanging information (the message) using face-to-face 
(F2F) communication.
We study message transmission and allow messages to be replicated (as opposed
to package delivery).
The goal is to give an algorithm which minimizes the time required to deliver 
the message from the source to the destination through a series of 
F2F message transfers.
In this paper we study how to complete this task efficiently and propose various 
centralized offline and distributed online algorithms which take into account 
the knowledge that the robots have about their speeds and initial locations.

\subsection{Model, Notation and Terminology}\label{sec:model}
The setup of our pony express problem will be in the Euclidean plane and points 
will be identified with their cartesian coordinates.
We use capital letters to denote points and lower-case letters with subscripts 
to denote their components (e.g. point $A = (a_1, a_2)$). 
For any points $A, B, C$, $|AB|$ denotes the Euclidean distance between $A$ and 
$B$, $\angle (ABC)$ denotes the angle formed by $A, B, C$ in this order, and 
$\triangle (ABC)$ denotes the triangle formed by $A, B, C$. 
Finally, $C(A, r)$ denotes a circle centered at $A$ with radius $r$.

Assume that $n$ robots $r_1,r_2,\ldots, r_n$ are placed at arbitrary positions 
in the Euclidean plane. 
The respective speeds of the robots are $v_1, v_2, \ldots,v_n$. 
The movement of a robot is determined by a well-defined trajectory. 
A robot trajectory is a continuous function $t \to f(t)$, with $f(t)$ the 
location of the robot at time $t$, such that $|f(t) - f(t')| \leq v |t-t'|$, 
for all $t, t'$, where $v$ is the speed of the robot. 
A robot can move with its own constant speed and during the traversal of its 
trajectory it may stop and/or change direction instantaneously and at any time.
Robot communication is F2F in that two robots can communicate (instantaneously) 
with each other only when they are co-located.

Algorithms describe the trajectories robots will follow and we will take into 
account the time it takes the algorithm to conclude the delivery task from the 
start, obtaining the message at a given source $S$, and eventually delivering 
it to a given destination $D$.
In general, we are interested in offline and online algorithms. 
In the offline setting, the locations and speeds of all robots are known in 
advance and are available to a central authority that assigns trajectories to 
the robots. 
In the online setting, the robots know only their own initial position and 
speed, along with the positions of $S$ and $D$. 
To measure the performance of our online algorithms, we consider their 
competitive ratio defined as follows. 
Let $t^*(I)$ be the optimal delivery time for an instance $I$ of a given 
problem and $t_A(I)$ be the time needed by some online algorithm $A$ for the 
same instance.
The competitive ratio of $A$ is $\max_I \frac{t_A(I)}{t^*(I)}.$
Our goal is to find online algorithms that minimize this competitive ratio.

\subsection{Related work}\label{sec:related}
Communicating mobile robots or agents have been used to address problems such 
as search, exploration, broadcasting and converge-casting, patrolling, 
connectivity maintenance, and area coverage (see \cite{chalopin2006mobile}). 
For example, \cite{beregrobustness} addresses the problem of how well a group 
of collaborating robots with limited communication range is able to monitor a 
given geographical space. 
To this end, they study broadcasting and coverage resilience, which refers to 
the minimum number of robots whose removal may disconnect the network and 
result in uncovered areas, respectively. 
Similarly, rendezvous is a relevant communication paradigm and in 
\cite{chuang2021,czyzowicz2020gathering} the authors investigate rendezvous 
in a continuous domain under the presence of spies. 
A related study on message transmission in a synchronized multi-robot system may be found in~\cite{beregrobustness}. 
Another application is patrolling whereby mobile robots are required to keep 
perpetually moving along a specified domain so as to minimize the time a point 
in the domain is left unvisited by an agent, e.g., 
see~\cite{czyzowicz2019patrolling} for a related survey.

Data delivery and converge-cast with energy exchange under a centralized 
scheduler were studied in~\cite{czyzowicz2016communication}. 
A restricted version concerns $n$ mobile agents of limited energy that are 
placed on a straight line and which need to collectively deliver a single piece 
of data from a given source point $S$ to a given target point $D$ on the line 
can be found in~\cite{chalopin2013data}. 
In \cite{chalopin2014data} it is shown that deciding whether the agents can 
deliver the data is (weakly) NP-complete. 
Additional research under various conditions and topological assumptions can be 
found in \cite{bartschi2017truthful} which studies the game-theoretic task of 
selecting mobile agents to deliver multiple items on a network and optimizing 
or approximating the total energy consumption over all selected agents, 
in~\cite{bartschi2016energy,bartschi2017energy,bilo2021new} which study data 
delivery and combine energy and time efficiency, and 
in~\cite{das2015collaborative,das2018collaborative} which are concerned with 
collaborative exploration in various topologies. 

Our problem was previously studied on graphs in~\cite{AnayaCCLPV16,Bartschi0M18,carvalho2019efficient,carvalho2021fast}. In particular it is shown in~\cite{carvalho2019efficient} that the problem can be solved with $k$ agents on an $n$-node, $m$-edge weighted graph in time $O(kn \log n + km)$. We use this algorithm in the development of our approximation algorithm.  

Our current work is related to the Pony Express communication problem  proposed in~\cite{coleman2021pony}. In that paper, the authors provide both optimal offline and online algorithms for the anycast and broadcast problems in the case where the underlying domain was a continuous line segment.  
To our knowledge, the planar case studied in our paper has not been considered previously.

\subsection{Outline and results of the paper}\label{sec:outline}
In Section~\ref{sec:offline_two_robots} we propose an optimal 
offline algorithm for two robots. 
For ease of exposition, we first consider the case when the slower robot
starts at the source and then the general case of arbitrary starting positions.
In Section~\ref{sec:sqrt2} we study the offline multirobot case. 
We propose an algorithm which approximates the optimal delivery time to within 
a factor of $\sqrt{2}$. 
Section~\ref{sec:Online Upper Bounds} is dedicated to online algorithms. 
For two robots we give an algorithm with competitive ratio of 
$\frac 17 \left( 5+ 4 \sqrt{2} \right)$ and show that for $n$ robots, this same 
algorithm has a competitive ratio of at most $2$. 
We also analyze lower bounds for this specific algorithm showing these bounds 
are tight. 
In Section~\ref{sec:lower_bounds_two_robots} we prove lower bounds on the 
competitive ratio of arbitrary online algorithms. 
We discuss two approaches, one where the position of a robot is unknown and the 
other where the speed of a robot is unknown. 
These different approaches provide lower bounds of 1.0391 and 1.0405, 
respectively. 
We conclude in Section~\ref{sec:conclusion} by discussing possibilities for 
additional research in this area. 

\section{Optimal Offline Algorithm for Two Robots}
\label{sec:offline_two_robots}

In this section, we will consider two robots $r_v$ and $r_1$ which can move with 
respective constant speeds $v$ and $1$ ($v > 1$) and design optimal offline 
algorithms with respect to the F2F communication model (observe that by scaling
the distances, setting the speed of the slow robot to be $1$ yields no loss of 
generality.)
Let $L$ and $K$ be the starting positions of robots $r_1$ and $r_v$, 
respectively.
There are three cases to consider:
\begin{enumerate}
    \item $\frac{|KS|}{v} \leq |LS|$: the fast robot can get to $S$ before the
        slow robot. 
        In this case, it is clear that in the optimal solution, the fast robot
        should move to $S$, acquire the message, and carry it to $D$.
    \item $\frac{|KD|}{v} \geq |LS| + |SD|$: the slow robot can deliver the 
        message to $D$ before the fast robot can even reach $D$. 
        In this case, the optimal solution is also clear.
        The slow robot should move to $S$, acquire the message, and carry it 
        to $D$.
    \item In all other cases, the slow robot can get to $S$ before the fast 
        robot, but the fast robot can get to the destination faster.
        The optimal solution, then, must involve a handover between the robots 
        at some point $M$ in the plane.
\end{enumerate}

For the first two cases, the optimal solution is trivial.
The third case, however, is not as we must find the point at which 
the robots meet.
First, we characterize the optimal meeting point $M$ for Case~3 through a series 
of lemmas.

\begin{lemma}\label{lm:optimal_m}
    For Case~$3$, there exists an optimal solution such that if $M$ is the 
    handover, then $|LS| + |SM| = \frac{|KM|}{v}$.
\end{lemma}
\begin{inlineproof}
    For the sake of contradiction, suppose $|LS| + |SM| \neq \frac{|KM|}{v}$.
    First, it is obvious that $r_1$ should move directly toward $S$ and then 
    directly toward $M$ and, similarly, $r_v$ should move directly toward $M$.
    Any other path could only increase the time to deliver the message.
    Then, since $|LS| + |SM| \neq \frac{|KM|}{v}$, either $r_1$ or $r_v$ must
    arrive at $M$ before the other. 
    Thus, there must be a time where one robot is waiting at $M$ for the other
    to arrive. 
    Let $t$ be the time that the first robot arrives to $M$ and $K^\prime$ be 
    the position of the other robot at time $t$. 
    We claim that an equal or better solution than waiting at $M$ would be for 
    the first robot to move along $K^\prime M$ until it meets the other robot at 
    some point $M^\prime$.
    Then the faster of the two robots carries the message from $M^\prime$ to 
    $D$ (Figure~\ref{fig:constant_movement}).
    \begin{figure} 
        \begin{center}
            \includegraphics[width=0.25\textwidth]{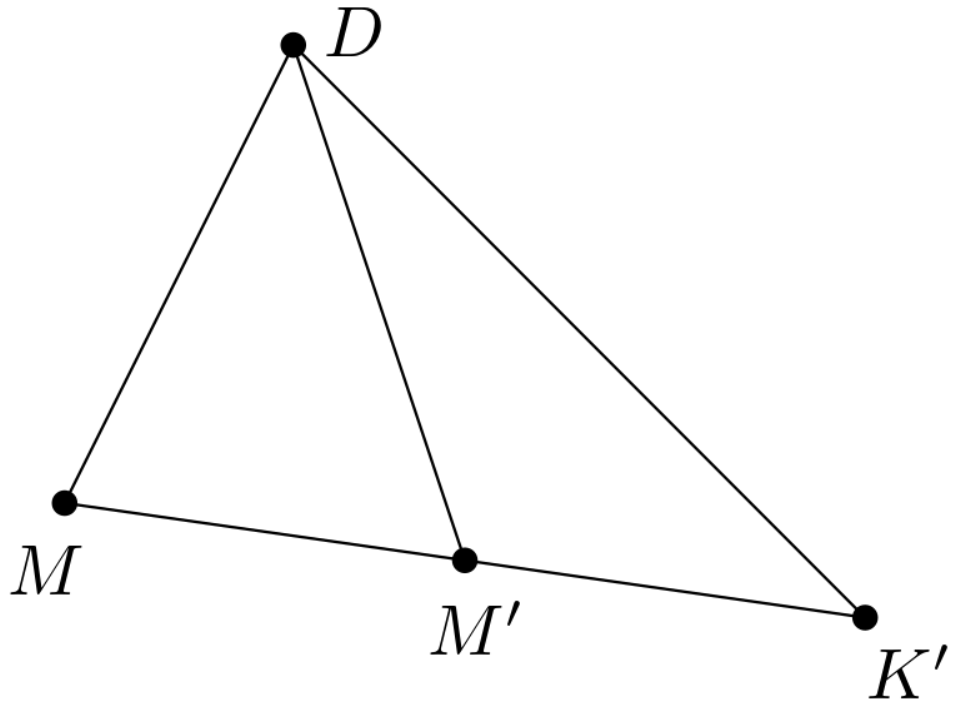}
        \end{center}
        \caption{The configuration at time $t$ where the first robot has arrived
        at $M$, the second is at position $K^\prime$. A better solution than 
        waiting for the second to arrive at $M$ would be for the robots to meet 
        at $M^\prime$.}
        \label{fig:constant_movement}
    \end{figure}
    If $r_1$ arrives at $M$ first, then if it waits at $M$ for $r_v$ 
    to arrive, the total delivery time is
    $t + \frac{|K^\prime M| + |MD|}{v}$, but since $\triangle K^\prime MD$ is 
    clearly larger (in perimeter) than $\triangle K^\prime M^\prime D$, then 
    \begin{align*}
        |K^\prime M^\prime| + |M^\prime D| \leq |K^\prime M| + |MD|
        \Rightarrow t + \frac{|K^\prime M^\prime| 
            + |M^\prime D|}{v} \leq t + \frac{|K^\prime M| + |MD|}{v}
    \end{align*}
    Thus, meeting at $M^\prime$ results in a quicker delivery, a contradiction
    to the assumption that $M$ is optimal.
    If $r_v$ arrives at $M$ first, then if it waits at $M$ for $r_1$ the 
    total delivery time is $t + |K^\prime M| + \frac{|MD|}{v}$, but 
    \begin{align*}
        |MM^\prime| + |M^\prime D| \leq |K^\prime M| + |MD| &\Rightarrow t 
            + \frac{|MM^\prime| + |M^\prime D|}{v} 
            \leq t + \frac{|K^\prime M| + |MD|}{v} \\ 
        &\Rightarrow t + |K^\prime M^\prime| + \frac{|M^\prime D|}{v} 
            \leq t + |K^\prime M| + \frac{|MD|}{v}
    \end{align*}
    Again, meeting at $M^\prime$ results in an equal or quicker delivery, and 
    so by contradiction, there must exist an optimal solution where 
    $|LS| + |SM| = \frac{|KM|}{v}$.
    \qed    
\end{inlineproof}
Intuitively, Lemma~\ref{lm:optimal_m} says that robots must move at their 
maximum speeds directly towards the location they will acquire the message and 
then directly toward the location they handover or deliver the message.
This restricts the set of feasible meeting points to the set of points in the 
plane such that both robots, moving in one direction at their maximum speeds,
meet at the same time.
For the case where the slow robot starts at the source ($L=S$), this is 
directly related to an ancient theorem by the Greek philosopher Apollonius,
which states ``the trajectory traced by a point $P$ which moves in such a 
way that its Euclidean distance from a given point $S$ is a constant multiple 
of its Euclidean distance from another point $K$ is a 
circle''~\cite{ogilvy1990excursions}. 
As a consequence, if the robots $r_1, r_v$ start at positions $S, K$, 
respectively, then the locus of points at which the two robots may travel 
directly towards and meet at the same time is the circle of Apollonius 
(see Figure~\ref{fig:circle1}). 
\begin{wrapfigure}{r}{0.4\textwidth}
    \centering
    \includegraphics[width=0.35\textwidth]{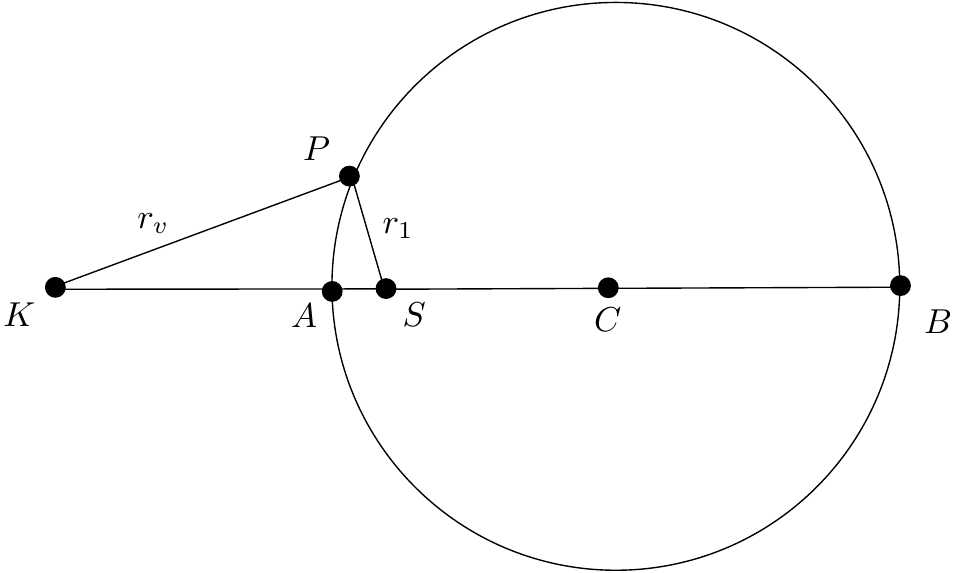}
    \caption{The Apollonius circle is the locus of points $P$ such that robots 
    $r_v$ and $r_1$ are equal time away from their starting positions $K$ and 
    $S$, respectively.}
    \label{fig:circle1}
\end{wrapfigure}
This circle, then is the locus of all possible handover points between the two 
robots.
The precise statement in the context of mobile robots is stated in 
Lemma~\ref{lm:apollonius}.

\begin{lemma}\label{lm:apollonius}
Two robots $r_v$ and $r_1$ with speeds $v$ and $1$ ($v>1$) are initially 
    placed at points $K$ and $S$, respectively. 
    The locus of points $P$ such that robots $r_v$ and $r_1$ are equal time 
    away from points $K$ and $S$, respectively, i.e., $\frac{|PK|}{|PS|} = v$, 
    forms a circle with center $C$ and radius $R$ so that
    \begin{align}\label{radius}
    C = S+ \frac{S -K}{v^2-1}    \mbox{ and }
    R = \frac{v |SK|}{v^2-1}
    \end{align}   
\end{lemma}
\begin{inlineproof}
    The proof follows easily by using the representation of the points $S, K, P$ 
    in cartesian coordinates and solving the equation $\frac{|PK|}{|PS|} = v$.
    \qed
\end{inlineproof}

The following definition of the Apollonius Circle will be used throughout this 
paper.
\begin{definition}[Apollonius Circle]
    The circle with center $C$ and radius $R$ given by Equations~\eqref{radius} 
    is called the Apollonius circle between robots $r_v$ and $r_1$ when their 
    respective starting positions are $K, S$.
\end{definition}

For instances of the problem where $L=S$ and whose optimal solutions do not 
involve either robot delivering the message by itself, the previous discussion
results in the following lemma whose proof follows directly from 
Lemmas~\ref{lm:optimal_m}~and~\ref{lm:apollonius}.
\begin{lemma}\label{lm:on_apollonius}
    The optimal meeting point $M$ is the point on the Apollonius circle between 
    robots $r_1$ and $r_v$ which minimizes the total delivery time
    $|SM| + \frac{|MD|}{v} = \frac{|KM| + |MD|}{v}$. 
    \qed
\end{lemma}

\subsection{Optimal algorithm when a robot starts at the source}
\label{sec:optimal_two_robots}

First we give an algorithm in the restricted case where one robot starts at the 
source where the message is located ($L=S$). 
\begin{figure}[!htb] 
    \centering
    \includegraphics[width=0.35\textwidth]{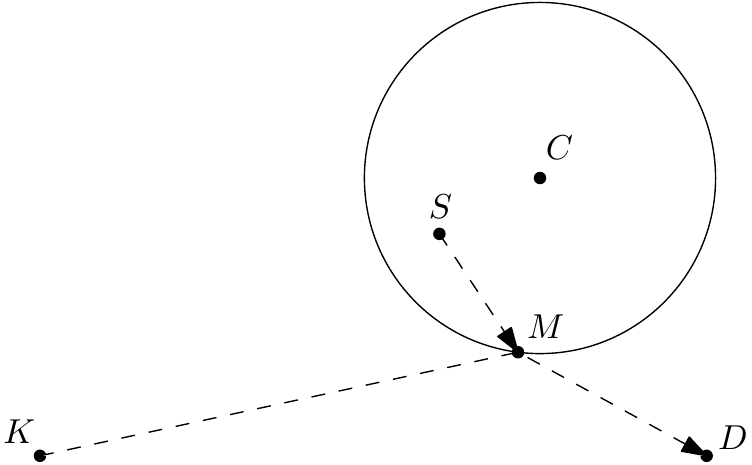} 
    \caption{The two-robot delivery problem with a slow robot at $S$, a faster 
    robot at $K$, and an Apollonius circle between the two centered at $C$.}
    \label{fig:op_offline_gen}
\end{figure}
Let $S=(0, 0)$ be the source, $K$ be the starting position of the fast robot
which we assume to be on the $x$ axis, and $D = (x, y)$ the destination. 
Without loss of generality, we assume $y \geq 0$ (if $y < 0$, the instance can
be reflected about the $x$ axis and solved equivalently, since $K$ is on the 
$x$ axis).
By Lemma~\ref{lm:on_apollonius}, our goal is to find the point $M$ on the 
robots' Apollonius circle which minimizes the delivery time 
$|SM| + \frac{|MD|}{v} = \frac{|KM| + |MD|}{v}$ (see~Figure~\ref{fig:op_offline_gen}).

Consider the following offline algorithm for computing the optimal delivery
time.

\begin{algorithm}[H]
    \caption{Optimal Two-Robot Algorithm with the 
    Slow Robot Starting at the Source}\label{alg:two_robot_source_opt}
    \begin{algorithmic}[1]
        \If {$\frac{|KD|}{v} \geq |SD|$}\label{line:case_slow}
            \State \textbf{return} $|SD|$
        \EndIf
        \State $\beta \gets \angle SKD$
		\If{$\sin(\beta) \leq \frac{1}{v}$}
  				\label{line:case_intersect}
				 \State $\alpha \gets \pi - \beta - \arcsin (v \sin \beta )$
        		\State $M \gets \frac{|SK|}{v^2-1} (v \cos \alpha-1, v \sin \alpha)$%
				\If{ $|KD| < |KM|$}
                     \State $M \gets $ point on Apollonius circle such that $CM$ bisects the  angle~$\angle (DMK)$\label{line:m_bisector}
				\EndIf
		\Else  
         
            \State $M \gets $ point on Apollonius circle such that $CM$ bisects the  angle~$\angle (DMK)$\label{line:m_bisector2}
		 \EndIf

        \State \textbf{return} $\frac{|KM| + |MD|}{v}$
    \end{algorithmic}
\end{algorithm}

\begin{theorem}
\label{thm:source}
    Algorithm~\ref{alg:two_robot_source_opt} returns the optimal delivery time 
    for instances with two robots $r_1$ and $r_v$ with speeds $1$ and $v$, and 
    starting positions $S$ and $K$, respectively.
    Algorithm~\ref{alg:two_robot_source_opt} can be implemented using a constant number of operations (including trigonometric functions).
\end{theorem}
\begin{inlineproof}
    First, note that Case~1 (from the three cases at the beginning of the 
    section) is not considered since the slow robot, $r_1$ is assumed to start 
    at the source.
    Case~2 is obviously handled by line~\ref{line:case_slow} in the algorithm.
    Case~3 is divided into two subcases based on whether or not the condition 
    in line~\ref{line:case_intersect} is satisfied.
    First, we consider the case where it is not.
    Let $\beta$ be the angle $\angle SKD$.
    Observe that if $KD$ is tangent to the Apollonius circle
    (Figure~\ref{fig:max_intersection_angle}), then
    $\sin(\beta) = \frac{|SK|v}{v^2-1}\cdot\frac{v^2-1}{|SK|v^2} = \frac{1}{v}$.
    Clearly for any smaller value for $\beta$, $KD$ intersects the Apollonius 
    circle at two points (and for any larger value, $KD$ does not intersect the 
    circle).
    
    \begin{figure}[!htb]
        \centering
        \includegraphics[width=0.35\textwidth]{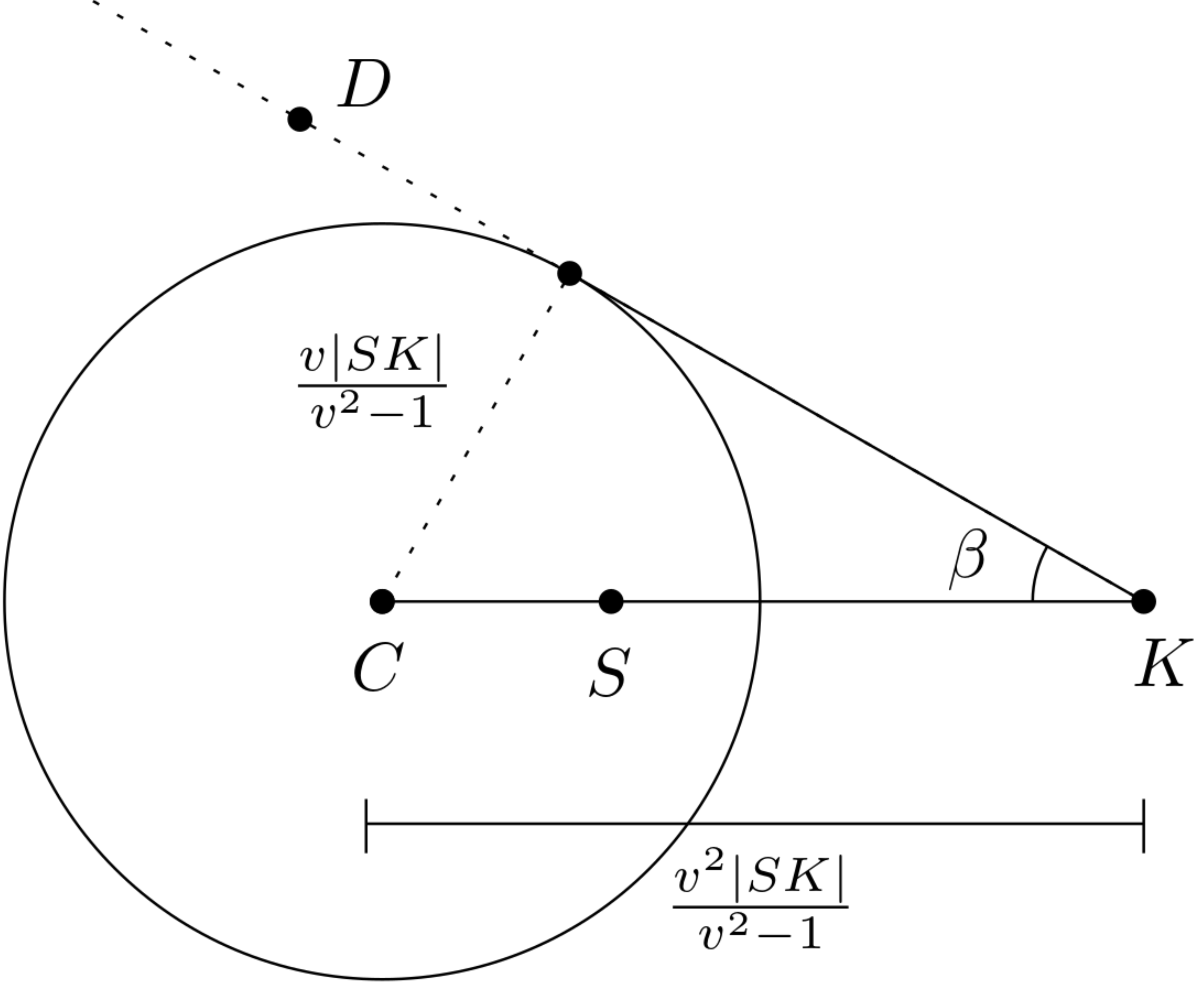}
        \caption{The maximum $\beta$ such that $KD$ intersects the Apollonius 
        Circle}
        \label{fig:max_intersection_angle}
    \end{figure}

    Then, let $\alpha = \angle KCM$ and $\gamma = \angle KMC$ 
    (Figure~\ref{fig:m_cases} left). By the law of sines
    $
        \frac{ (v^2 - 1)\sin \gamma }{|SK|v^2} 
            = \frac{ (v^2 - 1)\sin\beta}{|SK|v} \mbox{ and }
        \gamma = \arcsin(v \sin\beta) .
    $
    Thus $\alpha = \pi - \beta - \arcsin(v \sin\beta)$ and 
    $M = \frac{|SK|}{v^2-1} (v \cos\alpha - 1, v \sin\alpha)$ is just 
    the associated point on the Apollonius circle.

    Observe $M$ is the intersection point closest to $K$.
    Since the condition in line~\ref{line:case_intersect} is satisfied, $r_v$ 
    can move directly toward $D$ and, without veering from a direct path, meet 
    $r_1$ at $M$, acquire the message, and continue towards $D$ to deliver the 
    message.
    We know, since the first case was not satisfied, that $r_v$ can reach the 
    destination before $r_1$ can, so this is clearly the optimal trajectory.
    \begin{figure}[!htb]
        \centering 
        \includegraphics[width=0.56\textwidth]{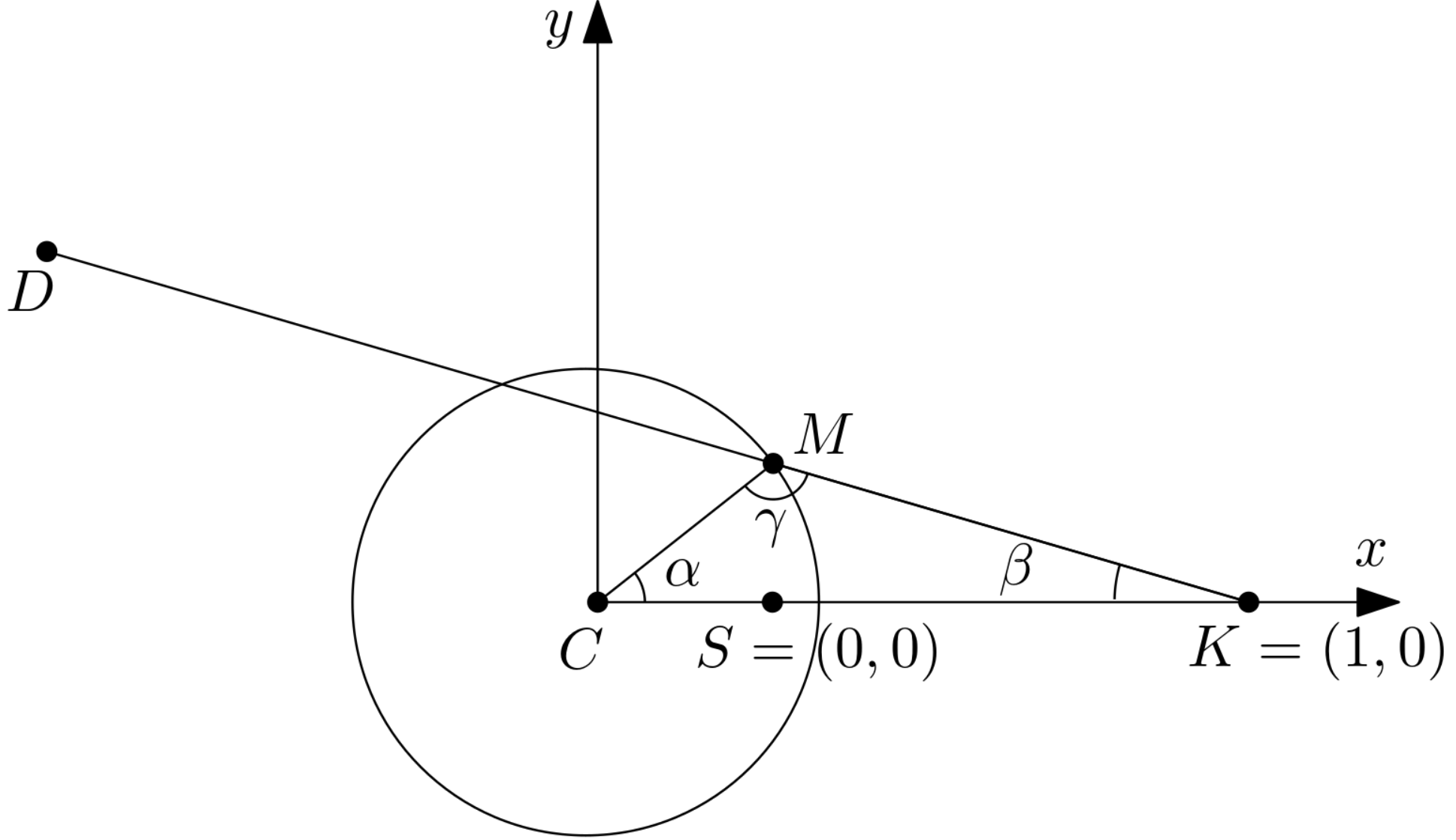}
        \includegraphics[width=0.43\textwidth]{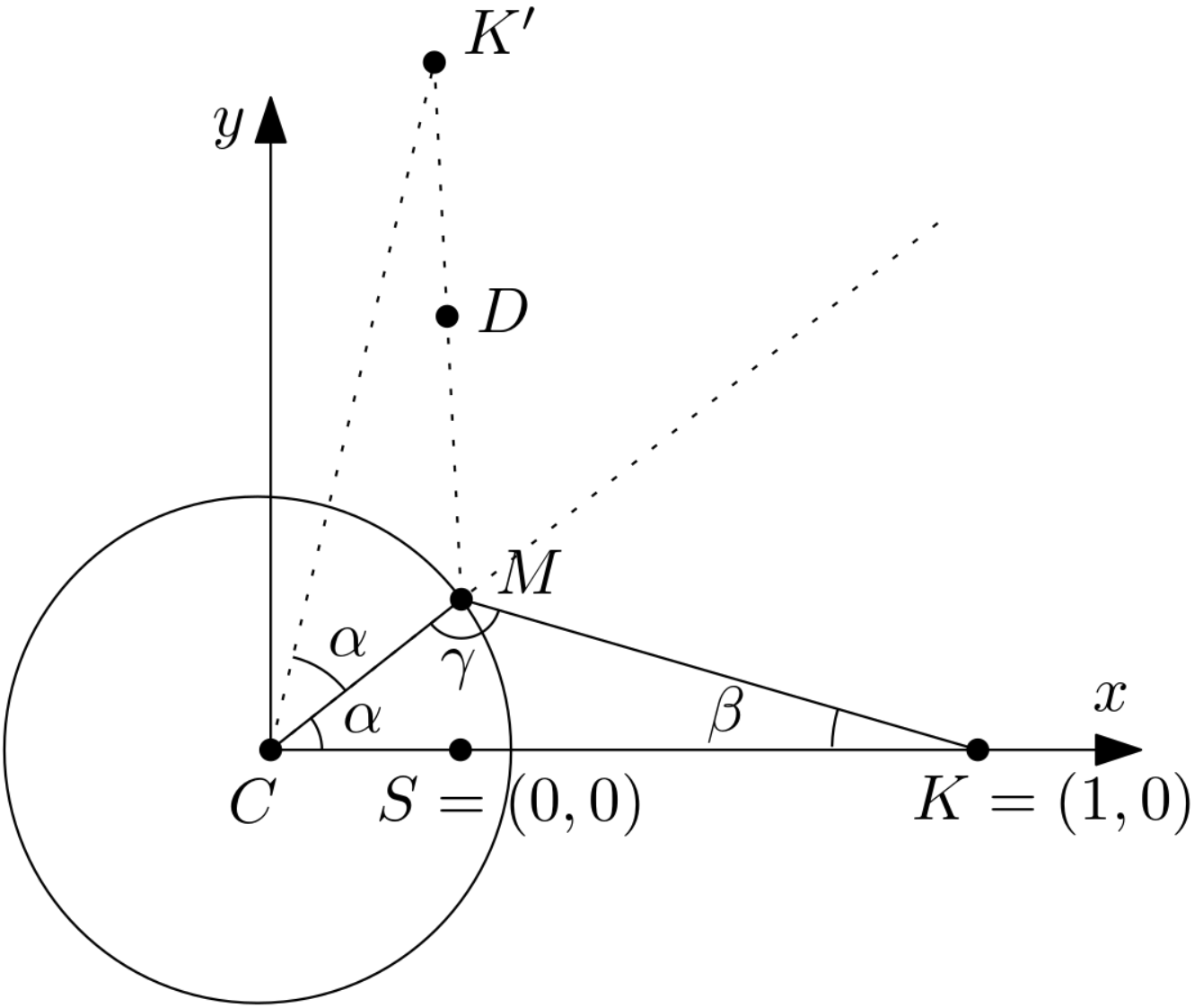}
        \caption{On the left, $D$ is such that the line-segment $KD$ intersects with the 
        Apollonius Circle. In the general case (right), $CM$ must bisect $\angle (DMK)$, and 
        thus $\alpha$ such that $MK^\prime$ be collinear to $MD$, where 
        $K^\prime = |CK|(\cos(\alpha), \sin(\alpha))$.}
        \label{fig:m_cases} 
    \end{figure}
    Now, suppose the condition in line~\ref{line:case_intersect} is not 
    satisfied.
    Consider the ellipse with foci $K$ and $D$ whose semi-major axis has length 
    $\frac 12 vt$ for some time $t \geq 0$.
    Then, by a defining property of an ellipse, the sum of the distances from 
    each foci to a point on the ellipse is equal to a constant value $vt$.
    Consequently, a robot starting at $K$ with speed $v$ takes exactly $t$ time 
    to travel to a point on the ellipse and then to $D$. 
    Observe that if the ellipse and Apollonius circle intersect, then the two 
    robots can meet at one of the intersection points and, by the previous 
    statement, the fast robot can deliver the message in time $t$.
    If they intersect at two points, though, then any point on the Apollonius 
    circle between these two intersections would yield a better solution. 
    The solution, then, is to find the minimal $t$ which causes the Apollonius 
    circle and the ellipse to intersect at exactly one point $M$. 
    Thus $CM$ must be normal to both the Apollonius circle and the ellipse. 
    That $CM$, therefore, must bisect $\angle (DMK)$ follows from a well-known 
    property of the ellipse, namely that a normal line through a point on an 
    ellipse bisects the angle it forms with the ellipse's foci.

    Next, we show the algorithm can be implemented to run using a constant number of operations (including trigonometric functions).
    The only lines in the algorithm that are not clearly computable with a constant number of operations 
    are lines~\ref{line:m_bisector} and \ref{line:m_bisector2}.
    To show that $M$ can be computed in constant time, we provide a formulation
    which can be given as input to Equation Solving tools (e.g.,  Mathematica)
    to find a closed-form solution~\footnote{
        \href{https://www.wolframcloud.com/obj/oscar.moralesponce/Published/Pony_Express_Theorem_1}{Link to Mathematica solution for Theorem~\ref{thm:source}}
    }.
    Let $\alpha = \angle KCM$ and $K^\prime$ be the point given by rotating 
    $K$ $2\alpha$ around $C$ (into the positive half-plane, 
    Figure~\ref{fig:m_cases} right).
    Observe that if $CM$ bisects $\angle (DMK)$, then $DK^\prime$
    is collinear with $MD$, or:
    $
        \frac{|SK| \cos(2\alpha) - |CS| - x}{|SK| \sin(2\alpha) - y}
        =  \frac{x - \cos\alpha - |CS|}{y - \sin\alpha}
    $        
where $D = (x,y)$.
    \qed
\end{inlineproof}

\subsection{Optimal algorithm in the general case}
In this subsection we consider the more general case where the slow robot does 
not start at the source. 
Let the starting positions of source and destination be $S = (s_1, s_2)$ and 
$D = (d_1, d_2)$ and let the robots $r_v$ and $r_1$ start from arbitrary points 
$K = (k_1, k_2)$ and $L = (l_1,l_2)$ in the plane, respectively. 
Again, we are interested in finding the point $M = (x_1,x_2)$ for the third case (from 
the cases at the beginning of the section), since optimal solutions for the 
first two cases are trivial to find.
As depicted in Figure~\ref{fig:circle3}, 
\begin{wrapfigure}{l}{0.5\textwidth} 
\begin{center}
\includegraphics[width=0.4\textwidth]{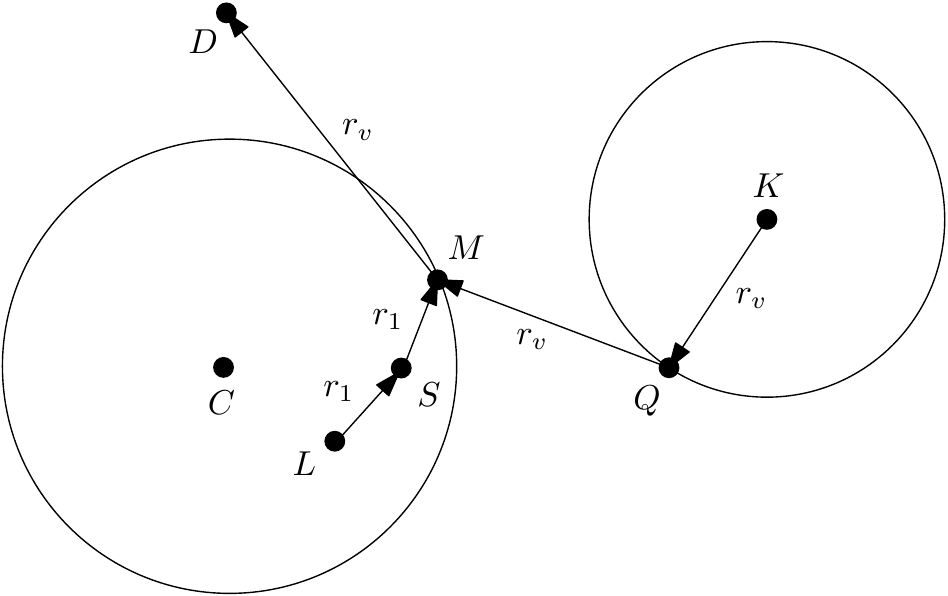}
\end{center}
\caption{%
    Trajectories of the robots for message delivery from $S$ to $D$. %
    Robot $r_v$ starts at the point $K$ and robot $r_1$ at the the point $L$. %
    Robot $r_1$ arrives at the source $S$ before $r_v$ does and meets robot %
    $r_v$ at $M$ which then delivers the message to $M$. %
}
\label{fig:circle3}
\end{wrapfigure}
robot $r_v$ follows a trajectory which first visits a point $Q$ at distance 
$v |LS|$ from its starting position, then continues along a straight-line 
trajectory to meet robot $r_1$ at a suitable point $M = (x_1,x_2)$, and finally 
delivers the message to the destination $D$. 
The main steps of the algorithm are as follows. 
\begin{enumerate}
\item
If $\frac{|KS|}{v} \leq |LS|$, then $r_v$ reaches $S$ before $r_1$ and 
$r_v$ should complete the delivery on its own.
\item
Otherwise, if $|LS| + |SD| \leq |KD|$, then $r_1$ can deliver the message on 
its own before $r_v$ can even reach the destination.
\item \label{item3}
Otherwise $r_1$ reaches $S$ in time $|LS|$ and, at the same time, robot $r_v$ 
goes to a specially selected point $Q = (q_1, q_2)$ which lies on the circle 
centered at $K$ with radius $v |LS|$.
\item \label{item4}
Robot $r_v$ meets robot $r_1$ at a point $M = (x_1, x_2)$ determined by the 
locus of points which are equal time away from $Q$ and $S$
(by Lemma~\ref{lm:apollonius}, this is the circle with center $C$ and radius 
$R$ as given in Equation~\eqref{radius}). Robot $r_1 $ passes message to $r_v$ 
which delivers it to $D$.
\end{enumerate}

Observe that by Lemma~\ref{lm:optimal_m}, $K$, $Q$, and $M$ must be collinear. 
We can then generalize the result of section~\ref{sec:optimal_two_robots}
using the following lemma. Recall that the center of similitude (also known as homothetic center) is a point from which at least two geometrically similar figures can be seen as a dilation or contraction of one another  (see \cite{yiu2001introduction}[Section 1.1.2]). 

\begin{lemma}
    \label{lm:similitude}
    Let $C$ be the center of the Apollonius circle between $r_1$ and $r_v$ when $r_1$ is at $S$ and $r_v$ is at $K$. 
    Then, $S$ is the center of similitude of the circles $\mathcal{C}(K, v|LS|)$ 
    and $\mathcal{C}(C, v|LS|/(v^2-1))$.
Consider any point $Q$ in the circumference of   $\mathcal{C}(K, v|LS|)$. Let $\beta$ be the angle $\angle(SKQ)$, then   
    $C_\beta = (\frac{v|LS|}{v^2-1} \cos \beta, %
        \frac{v|LS|}{v^2-1} \sin\beta) + C$ 
    is the center of the Apollonius circle of $S$.
\end{lemma}
\begin{inlineproof}
Refer to Figure~\ref{fig:similitude}.
    Consider any point $Q$ at the circumference of the circle 
    $\mathcal{C}(K, v|LS|)$.  
    Let $C_\beta$ be the center of the Apollonius circle between $r_1$ and $r_v$ when $r_1$ is at $S$ and $r_v$ is at $Q$. 
    Therefore, $|CQ| =  v^2 |SQ| / (v^2 - 1)$ and 
    $|SC_\beta| = |SQ| / (v^2 -1)$.
    Observe that the ratio $|Q|/|(SC_\beta| = v^2 - 1$ is always constant. 
    \begin{figure}[!ht]
        \centering 
        \includegraphics[scale=.8]{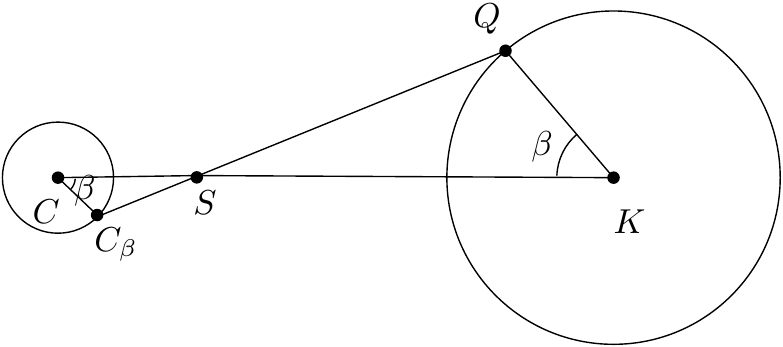}
        \caption{Similitude center.}
        \label{fig:similitude}
    \end{figure}
    Therefore, $C_\beta$ defines a circle $\mathcal{C}(C_\beta, v|LS|/(v^2-1))$ 
    where $S$ is the center of similitude. 
    The lemma follows since the triangles $\triangle(CSC_\beta)$ and 
    $\triangle(SKQ)$ are similar.
    \qed
\end{inlineproof}

Consider two points $C_\beta$ and $Q$ as described in Lemma~\ref{lm:similitude}, 
for some $\beta$. 
We can now use Theorem~\ref{thm:source} to characterize the solution.
However, this approach does not lead to a closed-form solution.
Instead, in the following lemma, we present another approach using optimization 
which does.

\begin{lemmarep}
    \label{lm:trajectory1}
    Let $a = |LS|$. 
    Then the optimal trajectory is obtained by a point $M = (x_1,x_2)$ which 
    minimizes the objective function 
    \begin{align}
        \begin{split}\label{eq:int0} 
    &  \sqrt{(k_1 - x_1)^2 + (k_2 - x_2)^2} 
            + \sqrt{(x_1 - d_1)^2 + (x_2 - d_2)^2}
        \end{split}
    \end{align}
    subject to the condition 
    \begin{align} 
        \label{eq:par3} 
        \left( 
                \frac{(x_1-k_1)^2 + (x_2-k_2)^2}{2a v^2}  
                - \frac{(x_1-s_1)^2 + (x_2-s_2)^2}{2 a} - \frac {a}{2} 
            \right)^2 
        = (x_1-s_1)^2 + (x_2-s_2)^2.
    \end{align}
\end{lemmarep}
\begin{appendixproof}
    Recall points $K, Q, M$ are collinear, the handover point  point $M$ must 
    lie at the intersection of two circles as depicted in 
    Figure~\ref{fig:circle4}.
    \begin{figure}[!htb]
        \centering
        \includegraphics[width=0.4\textwidth]{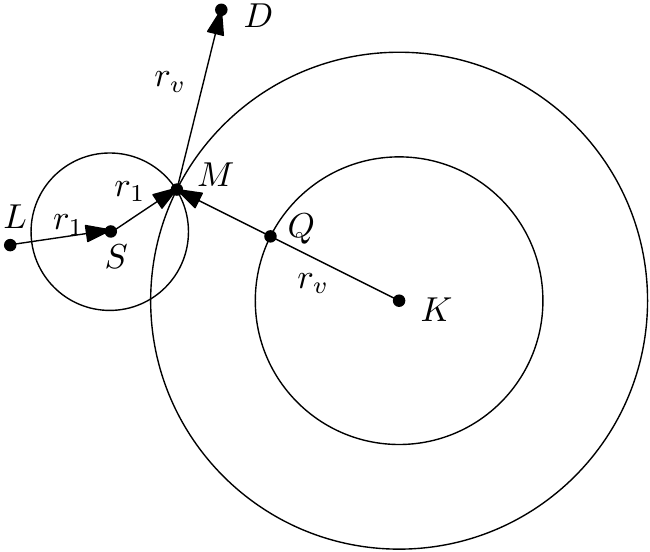}
        \caption{%
            The length of the trajectory of $r_v$ is minimized when the meeting%
            point $M$ is chosen so that the points $K, Q, M$ lie on a straight%
            line. %
        }
        \label{fig:circle4}
    \end{figure}
    This can be expressed by the fact that $M=(x_1,x_2)$ satisfies the two 
    equations
    \begin{align}
        &\label{eq:par1} (x_1-k_1)^2 + (x_2-k_2)^2 = v^2 (|LS| + t)^2 \\
        &\label{eq:par2} (x_1-s_1)^2 + (x_2-s_2)^2 = t^2 .
    \end{align}

    In turn, this gives a system of quadratic equations parametrized with 
    respect to time $t$. 
    We can rewrite Equation~\eqref{eq:par1} as 
    $\frac{(x_1-k_1)^2 + (x_2-k_2)^2}{v^2} = (|LS| + t)^2$
    and subtracting both sides of the last Equation from 
    Equation~\eqref{eq:par2} we derive the Equation
    $$
    \frac{(x_1-k_1)^2 + (x_2-k_2)^2}{v^2}
        - (x_1-s_1)^2 - (x_2-s_2)^2 = |LS| (|LS| + 2t) .
    $$

    By collecting similar terms, using Equation~\eqref{eq:par2}, and simplifying 
    we derive the following Equation
    \begin{align*} 
    \left( 
            \frac{(x_1-k_1)^2 + (x_2-k_2)^2}{2|LS| v^2}  
            - \frac{(x_1-s_1)^2 + (x_2-s_2)^2}{2 |LS|} - \frac {|LS|}{2} 
        \right)^2
    = (x_1-s_1)^2 + (x_2-s_2)^2.
    \end{align*}
    This is exactly Equation~\eqref{eq:par3}.
    \qed
\end{appendixproof}

The resulting optimization problem has two unknowns $x_1, x_2$ in the objective 
function~\eqref{eq:int0} and must satisfy the condition of 
Equation~\eqref{eq:par3}.
It can be used to substitute variables and express the final optimization 
function described in Formula~\eqref{eq:int0} using only a single variable, 
say $x_1$, which can then be minimized using standard analytical methods. 
This is easily seen since Equation~\eqref{eq:par3} is of degree $4$ in the 
variable $x_2$ (as well as in the variable $x_1$, for that matter).
Therefore a closed form expression of the variable $x_2$ in terms of the 
variable $x_1$ and the known parameters $S,D$ is easily derived. 

There are two symmetries in Equation~\eqref{eq:par3} which simplify the 
objective function and make the calculation of the solution easier. 
They are easily revealed with simple geometric transformations.

For the first symmetry, consider a rotation of the axis and a translation of 
the entire configuration of points so that $K$ and $S$ lie on the horizontal 
axis, i.e., $(k_1, k_2) = (0,0)$ and $(s_1,s_2) =(s_1,0)$. 
Then Equation~\eqref{eq:par3} is transformed to the equation
\begin{align} 
    \label{eq:par3a}  
    \left( \frac{x_1^2 + x_2^2}{2a v^2} - \frac{(x_1-s_1)^2 + x_2^2}{2a} 
        - \frac a2 \right)^2 = (x_1-s_1)^2 + x_2^2
\end{align}
The resulting symmetry is along the horizontal $x_1$-axis in 
Equation~\eqref{eq:par3a}. 
Namely, if $(x_1,x_2)$ is a solution so is $(x_1, - x_2)$. 
If we consider Equation~\eqref{eq:par3a} in the unknown $x_2$ we see that it is 
of degree $4$, but which is also a quadratic in $x_2^2$.  
Therefore $x_2$ can be easily expressed as a function of $x_1$ using the 
formula for the roots of the quadratic equation.
The second symmetry is obtained in a similar manner. 
If $(x_1,x_2)$ is a solution so is $(-x_1,  x_2)$. One considers a 
rotation of the axis and a translation of the entire configuration of points 
so that $K$ and $S$ lie on the vertical axis, i.e., $(k_1, k_2) = (0,0)$ 
and $(s_1,s_2) =(0,s_2)$. Details can be completed as above.
To sum up we have the following Algorithm~\ref{alg:two_robot_opt} which 
determines the handover point which yields the optimal trajectory.
\begin{algorithm}[H]
    \caption{Optimal Two-Robot Algorithm}\label{alg:two_robot_opt}
    \begin{algorithmic}[1]
        \If{$\frac{|KS|}{v} \leq |LS|$}\label{line:gen_case_fast}
            \State \textbf{return} $\frac{|KS| + |KD|}{v}$
        \ElsIf {$\frac{|KD|}{v} \geq |SD|$}\label{line:gen_case_slow}
            \State \textbf{return} $|SD|$
        \Else
            \State {%
                $M^* \gets M$ which minimizes %
                Formula~\eqref{eq:int0}~\label{line:opt_m}%
            }
            \State \textbf{return} $\frac{|KM^*|+|M^* D|}{v}$
        \EndIf
    \end{algorithmic}
\end{algorithm}

\begin{theoremrep}       
    \label{thm:optimal_trajectory_2}
    Algorithm~\ref{alg:two_robot_opt} returns the optimal delivery time for two 
    robots $r_1$ and $r_v$ with speeds $1$ and $v$, respectively, and can be 
    implemented using a constant number of operations (including trigonometric functions).
\end{theoremrep}
\begin{appendixproof}
    The proof follows from the previous discussion. 
    Indeed, without loss of generality we may consider only the case where the 
    slow robot $r_1$ reaches $S$ first. 
    As depicted in Figure~\ref{fig:circle4} there are two competing 
    trajectories.
    Given that the slow robot can arrive first at $S$, either the robots follow 
    the algorithm and the slow robot $r_1$ meets the faster robot $r_v$ at the 
    meeting point $M$ to handover the message to $r_v$ which then delivers it to 
    $D$ or the faster robot $r_v$ gets the message at $S$ and delivers it to 
    $D$ without cooperating with the other robot. 

    In the former case the delivery time will be $\frac{|KM|}v + \frac{|MD|}v$ 
    while in the latter case the delivery time will be $\frac{|KS| + |SD|}v$. 
    However this is easy to prove since the point $M$ must lie inside the 
    triangle $\triangle (KSD) $ as depicted in Figure~\ref{fig:circle4}, i.e.,
    $|KM| + |MD| < |KS| + |SD|$.
    All other lines in the algorithm clearly require a constant number of operations. 
    By the previous discussion, a closed form solution for the optimization 
    required in line~\ref{line:opt_m} exists~\footnote{
        \href{https://www.wolframcloud.com/obj/oscar.moralesponce/Published/Pony_Express_Theorem_2}{Link to Mathematica solution for Theorem~\ref{thm:optimal_trajectory_2}}
    }. 
    \qed
\end{appendixproof}

\section{Offline $\sqrt{2}$ Approximation for Multiple Robots}
\label{sec:sqrt2}

\begin{wrapfigure}{r}{0.5\textwidth}
    \begin{center}
        \includegraphics[width=0.15\textwidth]{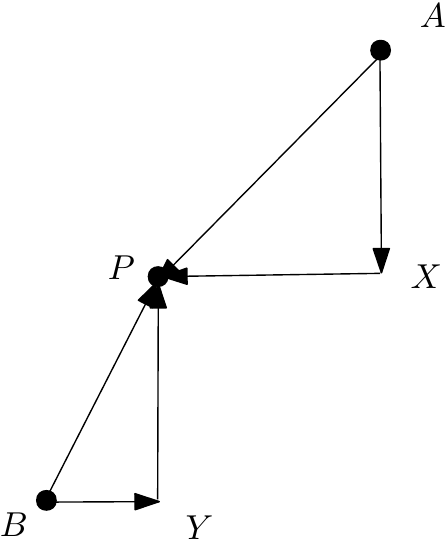}
    \end{center}
    \caption{%
        Replacing Euclidean movements of the robots with rectilinear movements.%
    }
    \label{fig:rect1}
\end{wrapfigure}
\vspace{-0.3cm}
In principle, the equations derived in the previous section can be generalized 
to solve the problem optimally for $n$ robots. 
Unfortunately, we are not able to solve the resulting set of equations. 
We do not speculate on the complexity of the general problem here. 
Instead, in this section we provide a $\sqrt{2}$-approximation algorithm, 
The robots know the location of the source $S$ and destination $D$ but also all 
robots know the initial locations and speeds of all other robots. 
The basic idea of our proof is contained in the following observation depicted 
in Figure~\ref{fig:rect1}.
Suppose that during the execution of an optimal ``Euclidean'' algorithm (i.e., 
optimal in the sense of the Euclidean distance) two robots placed at $A$ and 
$B$, follow the straight-line trajectories $A \to P$ and $B \to P$, 
respectively, and meet at the point $P$. 

Now we replace the Euclidean trajectories $A \to P$ and $B \to P$ with the 
rectilinear trajectories $A \to X \to P$ and $B \to Y \to P$, respectively. 
Elementary geometry implies that 
\begin{equation}
    \label{eq:rect1}
    |AX| + |XP| \leq \sqrt{2} |AP|
    \mbox{ and } 
    |BY| + |YP| \leq \sqrt{2} |BP| .
\end{equation}
This observation leads to the following lemma.
\begin{lemma}
    \label{lm:sqrt2}
    Consider the pony express problem for $n$ robots, a source $S$ and a 
    destination $D$ in the plane. Then
    $Opt_{Rect} \leq \sqrt{2} \cdot Opt_{Eucl}$, where $Opt_{Rect} , Opt_{Eucl}$ 
    are the delivery costs of the optimal trajectories of the pony express 
    problem for delivering from a source to a destination measured in the 
    rectilinear and Euclidean metrics, respectively.  
\end{lemma}
\begin{proof}
    Consider an optimal Euclidean algorithm ${\cal A}_{Eucl}$ which ensures 
    the delivery time is exactly $Opt_{Eucl}$, i.e., 
    $T({\cal A}_{Eucl}) = Opt_{Eucl}$. 
    Now use the idea discussed in Figure~\ref{fig:rect1} to replace the 
    Euclidean trajectory of algorithm ${\cal A}_{Eucl}$ with a rectilinear 
    trajectory thus giving rise to a rectilinear algorithm ${\cal A}_{Rect}$. 
    Note that in this rectilinear simulation of the optimal Euclidean solution, 
    robots may not arrive at the endpoints at the same time. 
    The robot that arrives first, should simply wait at the meeting point until 
    the second robot arrives.
    The meeting time is thus determined by the last arrival of the two robots.
    By definition, the time it takes the rectilinear algorithm ${\cal A}_{Rect}$ 
    to deliver the message is at least $Opt_{Rect}$, i.e., 
    $T({\cal A}_{Rect}) \geq Opt_{Rect}$. 
    From Inequality~\eqref{eq:rect1} we have that 
    $T({\cal A}_{Rect}) \leq \sqrt{2} \cdot T({\cal A}_{Eucl})$. 
    Therefore we conclude that
    $
        Opt_{Rect} \leq T({\cal A}_{Rect}) 
            \leq \sqrt{2} \cdot T({\cal A}_{Eucl}) 
            = \sqrt{2} \cdot Opt_{Eucl}
    $
    \qed
\end{proof}

Consider $n$ robots in the plane with starting positions $p_1, \ldots, p_n$. Without loss of generality assume the slowest robot has speed $1$. Further, let the source of a message be located at a point $S$ and the destination at a point $D$ and assume, without loss of generality, that the line segment $SD$ is horizontal. Enclose the points $S, D$ and $p_1, \ldots, p_n$ in a $\Delta \times \Delta$ square with sides parallel to the $x-, y-$axis, where $\Delta $ is a positive real proportional to the diameter of the set $\{ S, D\} \cup \{ p_1, \ldots, p_n \}$. For $\epsilon >0$ arbitrarily small, partition the $\Delta \times \Delta$ square with parallel vertical and horizontal lines with consecutive distances $\epsilon >0$, respectively, so as to form a $\frac{\Delta}{\epsilon} \times \frac{\Delta}{\epsilon}$ grid. Without loss of generality we may assume that $S$ and $D$ are vertices in this grid graph (This is easy to accomplish by choosing $\epsilon$ to be an integral fraction of the distance $|SD|$ between $S$ and $D$.)  Clearly, this forms a grid graph with $\left( \frac{\Delta}{\epsilon} \right)^2$ vertices so that $S, D$ are also vertices and $\left( \frac{\Delta}{\epsilon} \right)^2$ edges. Now consider the following algorithm.


\vspace{-0.3cm}
\begin{algorithm}[H]
\caption{Grid Algorithm ($S$ source, $D$ destination, $\epsilon > 0$)}\label{alg:grid1}
\begin{algorithmic}[1]
\State {In phase 1, each robot moves from its starting position $p_i$ to one vertex $p_i'$ of the $\epsilon \times \epsilon$ square in which it is contained; all the robots synchronize so that they can start the next phase at the same time by waiting for time at most $\epsilon$;}
\State {In phase 2, run the optimal algorithm in \cite{carvalho2019efficient} on the $\frac{\Delta}{\epsilon} \times \frac{\Delta}{\epsilon}$ grid to provide trajectories for the $n$ robots with starting positions $p_1', \ldots, p_n'$ in optimal time in order to deliver the message from the course $S$ to the destination $D$; when a robot meets another robot for a message handover the robot that arrives first waits for the arrival of the second robot;}
\end{algorithmic}
\end{algorithm}

\begin{theorem}
\label{thm:grid}
For any $\epsilon' > 0$ arbitrarily small, there exists an algorithm that finds trajectories for $n$ robots to deliver the message from the source $S$ to the destination $D$ in time $O\left( n^3 \left( \frac {\Delta}{\epsilon '} \right)^2 \log \left(n \frac {\Delta}{\epsilon '} \right)  \right) $ whose delivery time is at most $\sqrt 2$ multiplied by the delivery time of the optimal Euclidean algorithm plus the additional additive overhead $\epsilon'$, where $\Delta$ is the diameter of the point set
. 
\end{theorem}
\begin{proof}
    Let $\epsilon = |SD| / \left\lceil \frac{n}{\epsilon^\prime} |SD| \right\rceil \leq \frac{\epsilon^\prime}{n}$ and run Algorithm~\ref{alg:grid1}. 
    Let ${\cal A}(p_1, \ldots, p_n)$ be the time the algorithm takes to deliver the message and let $Grid(p_1^\prime, \ldots, p_n^\prime)$ be the time for step 2 in the algorithm (the optimal delivery time for the given grid with starting positions $p_1^\prime, \ldots, p_n^\prime$). 
    Then, let $Opt_{Rect}$ and $Opt_{Eucl}$ be the optimal delivery times for the rectilinear and Euclidean metrics respectively.
    First, observe 
    \begin{equation}\label{obs1}
        {\cal A}(p_1, \ldots, p_n) \leq Grid(p_1^\prime, \ldots, p_n^\prime) + \epsilon
    \end{equation}
    The result will follow from Lemma~\ref{lm:sqrt2} and the following claim:

\begin{claim}
$Grid(p_1^\prime, \ldots, p_n^\prime) \leq Opt_{Rect} + (n-1)\epsilon$.
\end{claim}     
    \begin{proof}(of Claim)
        Without loss of generality, assume the robots involved (in order) are robots $1$ through $k$.
        Let $q_1, q_2, \ldots, q_k$ be the handover points in the optimal rectilinear algorithm (where $q_1=S$). 
        Let $q_i^\prime$ be the nearest point to $q_i$ on the grid.
        Let $t_0$ be the time it takes for the first robot to arrive at the source and, $t_i$ (for $i\in[1,k]$) be the time that robot $i$ holds the message.
    
        Consider the algorithm where robots emulate the rectilinear algorithm on the grid by meeting 
        at $q_i^\prime$ instead of $q_i$, for each handover.
     Note that robots may not arrive at the endpoints at the same time. However, the algorithm is offline the robot that arrives first, should simply wait at the meeting point until the second robot arrives. 
        
        The time each robot holds the message then is at most $t_i + \epsilon$ and the total time to deliver the message is 
        \begin{align*}
            \sum_{i=1}^{k} t_i + \epsilon = Opt_{Rect} + (k-1)\epsilon \leq Opt_{Rect} + (n-1) \epsilon.
        \end{align*}
This proves the claim.        
    \end{proof}



    Using inequality~\eqref{obs1}, the above Claim and Lemma~\ref{lm:sqrt2} we arrive at the inequality
    $
        {\cal A}(p_1, \ldots, p_n) \leq \sqrt 2  \cdot Opt_{Eucl} + n\epsilon.
    $
    Finally, by selecting $\epsilon$ as above, ${\cal A}(p_1', \ldots, p_n') \leq \sqrt{2} \cdot Opt_{Eucl} + \epsilon^\prime$ and the complexity of the algorithm, by \cite{carvalho2019efficient}, is $O\left( n^3 \left( \frac {\Delta}{\epsilon '} \right)^2 \log \left(n \frac {\Delta}{\epsilon '} \right)  \right) $. This completes the proof of Theorem~\ref{thm:grid}.
    \qed
\end{proof}

\section{Online Upper Bounds}
\label{sec:Online Upper Bounds}

In this section we discuss online algorithms. In Subsection~\ref{sec:cr2D} we 
give an online algorithm with competitive ratio  $\frac{1}{7}(5+4 \sqrt{2})$ 
for two robots with knowledge only of the source $S$ and destination $D$. 
We show this bound is tight for the given algorithm. 
In Subsection~\ref{Multi Robot Algorithm} we show that the same algorithm when 
generalized to $n$ robots has competitive ratio at most 2. 
Further, we show that for any $n>2$, the competitive ratio of our algorithm is 
at least $2 - \frac{2}{2^{n} - 1}$.

\subsection{%
    Two Robot Algorithm with Competitive Ratio $\frac{1}{7}(5+4 \sqrt{2})$ %
}\label{sec:cr2D}

Consider the following Algorithm~\ref{alg:approx1} for multiple robots.

\vspace{-0.3cm}
\begin{algorithm}[H]
    \caption{Online Algorithm ($S$ source, $D$ destination)}\label{alg:approx1}
    \begin{algorithmic}[1]
        \State {Move toward $S$} 
        \State {Acquire the message at $S$}
        \State {Move toward and deliver the message to $D$}
    \end{algorithmic}
\end{algorithm}
\vspace{-0.3cm}

Observe that in this algorithm the robots act independently. 
In particular no attempt is made to co-ordinate the action of the robots and if 
two robots meet they ignore each other. 
This is not required in order to achieve the upper bounds below. 
For our lower bounds on this algorithm, we assume that the robots do not 
interact even if it may improve the time to complete the task. 

\begin{theorem}\label{thm:approx-new2D}
    For the case of two robots, Algorithm~\ref{alg:approx1} delivers the message 
    from the source $S$ to the destination $D$ in at most 
    $\frac{1}{7}(5+4 \sqrt{2})$
    times the optimal offline time.
\end{theorem}
\begin{inlineproof} 
    Given an arbitrary instance of the problem, let $t^*$ be the time taken by 
    the optimal solution to deliver the message from $S$ to $D$. 
    Let $r_1, r_2$ be the robots involved in that optimal solution where 
    $v_1 = 1  \leq v_2 = v$. 
    Observe that if only one robot is involved then our online solution is 
    optimal. 
    Let $K$ be the starting point of robot $r_2$. 
    Let $M$ be the point in the optimal solution where $r_1$ hands the message 
    off to $r_2$. Let $t$ be the time when this happens. 
    Finally, let $x = |SM|$ and $a=|SD|$.

    We make the following observations: 
    \begin{enumerate}
        \item $\frac{|KM|}{v} = t$.
        \item $t^* = \frac{|KM|+|MD|}{v} = t + |MD|/v$.
        \item $r_1$ delivers the message in time $t+a-x$. (Recall that $t$ is 
            the time for $r_1$ to reach $M$, $x = |SM|$, $|SD| = a$ 
            and $v_1 = 1$.)
        \item $r_2$ delivers the message in time $\frac{|KS|+|SD|}{v}$ which is 
            less than $\frac{|KM| +|SM|+a}{v} = t + \frac{a+x}{v}$.
        \item $t \geq x$ and  $x \leq a$.
        \item $|MD| \geq a-x$.
    \end{enumerate}

    Let $c_2$ be the competitive ratio of our algorithm for our two robots. 
    We have 
    $
        c_2 = \min \left\{ 
            \frac {t+a-x}{t^*}, \frac{\frac{|KS|+|SD|}{v}}{t^*} 
        \right\} \leq 
        \min \left\{
            \frac{t+a-x}{t+|MD|/v},\frac{ t+\frac{a+x}{v}}{t+|MD|/v} 
        \right\}
    $ by observation 4.
    This is maximized when $t+a-x = t+\frac{a+x}{v}$ or when 
    $x = a \frac{v-1}{v+1}$. 
    In this case we have:
    \begin{align*}
        c_2 & \leq  \frac{t+a-x}{t+|MD|/v} \\
        & \leq \frac{t+a-x}{t + (a-x)/v}  \textrm{ by obs. 6 }\\
        & \leq \frac{x+a-x}{x + (a-x)/v} \textrm{ by obs. 5 }\\
        & = \frac{v^2+v}{v^2-v+2}.
    \end{align*}

    This is maximized when $v = 1+\sqrt{2}$ at which point 
    $c_2 \leq \frac{1}{7}(5+4 \sqrt{2}).$
    \qed
\end{inlineproof}

\begin{example} 
\label{ex:2}
Now we give a tight lower bound on the competitive ratio of 
Algorithm~\ref{alg:approx1} for two robots.
Consider the following example input.
One robot is placed at the source $S$ which is the point $(0,0)$ and has 
speed $\frac 1{1+\sqrt{2}}$. The destination $D$ is placed at the point $(1,0)$. 
The second robot has speed $1$ and is placed at the point $(\sqrt{2}, 0)$.  
The robots are initially placed at distance $\sqrt{2}$. 
In the optimal algorithm the robots meet in time 
$\frac{\sqrt{2}}{1+ \frac 1{1+\sqrt{2}}} = 1$ at the point 
$x = \frac 1{1+\sqrt{2}}$. 
The faster robot picks up the message at $x$ and delivers it to $D$ in 
additional time $1-x = 1 -  \frac 1{1+\sqrt{2}}$. 
Therefore the delivery time of the optimal algorithm is equal to
$1+ 1-x = 2 -  \frac 1{1+\sqrt{2}}$.
It follows that the competitive ratio satisfies
$
    c_2 \geq \frac{1+\sqrt{2}}{2 -  \frac 1{1+\sqrt{2}}} 
        = \frac{1}{7}(5+4 \sqrt{2}) \approx 1.522407\ldots.
$
\end{example}

\begin{remark}
    Note that in Example~\ref{ex:2} if we parametrize the speed of the slow 
    robot to $\frac 1{1+y}$, and place the fast robot at position $y >1$ then 
    similar calculations show that $c_2 \geq \frac{y^2+3y+2}{y^2+y+2}$. 
    Further, it is easy to see that the lower bound $\frac{y^2+3y+2}{y^2+y+2}$ 
    is maximized for $y = \sqrt{2}$.
\end{remark}

\subsection{Multi Robot Algorithm with Competitive Ratio $\leq 2$}
\label{Multi Robot Algorithm}
\begin{theorem}\label{thm:pony_online}
    Algorithm~\ref{alg:approx1} has competitive ratio at most $2$ for any $n>2$ 
    robots.
\end{theorem}
\begin{inlineproof}
    To simplify notation, let the distance between the source and destination 
    be $1$.
    Observe that in the algorithm, every robot attempts to deliver the message 
    entirely by itself. The robots do not cooperate at all. 
    Clearly, then, if the robots can optimally deliver the message in $t^*$, 
    then Theorem~\ref{thm:pony_online} holds if and only if there exists at least 
    one robot that can deliver the message by itself in $2t^*$ time. 
    In other words, we must show that there is a robot $i$ with speed $v_i$ and 
    that starts a distance $s_i$ from the source such that 
    $\frac{1 + s_i}{v_i} \leq 2t^*$ ($\exists i$ such that 
    $v_i \geq \frac{1 + s_i}{2t^*}$).

    For the sake of contradiction, assume all robots have speed 
    $v_i < \frac{1 + s_i}{2t^*}$. 
    Then we must show that they could not possibly deliver the message optimally 
    in time $t^*$ or less. 
    By restricting the robots speed as a function of their distance to the 
    source, we allow the robots to choose everything else about their starting 
    positions to minimize the optimal delivery time.
    Clearly, robots will most quickly deliver the message if they are positioned 
    on the line from $[0, \infty)$ (robots are always moving directly toward the 
    message or its destination).

    Furthermore, the more robots in the system, the faster the message will be 
    delivered (observe that given two participating robots, inserting an 
    additional robot between them improves the delivery time).
    Therefore, we can assume there are an uncountably infinite number of robots 
    (one at every point on the interval $[0, \infty)$) and the problem becomes 
    continuous.

    Consider a robot that starts at position $s$ on the line.
    By construction, its velocity is less than $\frac{1+s}{2t^*}$ and after time 
    $t$ its position is $x > s - \frac{1 + s}{2t^*} t$.
    Therefore, if the message is at position $x$ on the line segment at time 
    $t$, its velocity at that moment is less than
    $
        \frac{1 + \frac{2 t^* x + t}{2t^* - t}}{2t^*} = \frac{1 + x}{2t^* - t}
    $
    since that is the upper bound on the speed of the fastest robot that could 
    reach $x$ by time $t$.

    It follows from the previous discussion that the speed of the message must 
    satisfy the inequality:
    $
        \frac{d~x(t)}{dt} < \frac{1 + x(t)}{2t^* - t} .
    $
    The resulting differential equation with unknown $x(t)$ and the initial 
    condition $x(0) = 0$, yields
    $
        x(t) < \frac{t}{2t^* - t} .
    $

    Finally, we can use this equation to show that the delivery time (when 
    $x(t)=1$) must be greater than $t^*$.
    Observe 
    $\frac{t}{2t^* - t} > 1 \Rightarrow t > 2t^* - t \Rightarrow t > t^*$,
    so the robots \textit{cannot} deliver the message in time $t^*$, a 
    contradiction.
    Therefore, if the optimal delivery is $t^*$, then there must exist a robot 
    $i$ with speed $v_i \geq \frac{1+s_i}{2t^*} \geq \frac{1+s_i}{2t^*}$ which 
    can deliver the message to the destination in at most $2t^*$ time. \qed
\end{inlineproof}

\begin{theorem}
    Given $n>2$ robots, there is a robot deployment such that 
    Algorithm~\ref{alg:approx1} has competitive ratio at least 
    $2 -\frac{2}{2^{n} -1} $.
\end{theorem}
\begin{inlineproof}
    Let $S$ be the source and $D$ the destination. 
    Without loss of generality, we assume $SD$ is a unit line where $S$ is at 
    the origin and $D = (1,0)$.
    Given the meeting points, we construct an instance with $n$ robots with 
    competitive ratio $2\left(1-\frac{1}{2^{n} -1} \right)$.
    Let $m_{i} = 1 - \frac{2^{i+1}-1}{2^{n} - 1}$ be the meeting point of robot 
    $r_i$ and $r_{i+1}$ for $i \in [0, n-2]$. 
    In our construction, robots are  required to arrive at $D$ at the same time 
    after reaching $S$. Hence, we compute the speed from the meeting points as 
    follows:
    Let $t$ be the time that robot $r_{i}$ arrives at point $m_i$. 
    Therefore, it arrives at $D$ after reaching $S$ at time 
    $t + \left(2 - \frac{2^{i+1}-1}{2^{n} - 1}\right)/v_i$. 
    Further, at time $t$ robot $r_{i+1}$ has reached 
    $\left( \frac{m_{i+1} - 2^{i+1}}{2^{n} - 1} \right) / v_{i+1}$. 
    Therefore, it arrives at $D$ after reaching $S$ at time 
    $t + \left(2 -   \frac{2^{i+2}-1 - 2^{i+1}}{2^{n} - 1}\right)/v_{i+1}$. 
    Setting both equal and factorizing we obtain the speed $v_{i+1}$ given 
    $v_i$, i.e. 
    $v_{i+1} = v_i \left(1 - \frac{2^{i+2}}{2^{n+1} -2^{i+1} -3} \right)$. 
    Initially, we set $v_0 = 1$ which implies that $r_0$ is the fastest robot. 
    Let $p_i = v_i \cdot 4 - 1$ be the initial position or robot $r_i$ on the 
    line $SD$. 
    Observe that all robots are in the interval $[-1, 3]$. 
    We claim that the competitive ratio of the setting is 
    $2\left(1 - \frac{1}{2^{n} - 1}\right)$. 
    Observe that $r_0$ arrives at $m_0$ at time $2 + \frac{1}{2^{n} - 1}$. 
    Therefore, in the optimal algorithm it takes an additional 
    $\frac{1}{2^{n} - 1}$ to reach $D$ meanwhile in Algorithm~\ref{alg:approx1} 
    robot $r_0$ takes $4$ to reach $D$. 
    Therefore, the competitive ratio is $4/\left(2+\frac{2}{2^{n}-1}\right)$. 
    Simplifying we obtain $2 -\frac{2}{2^{n} -1}$ and the theorem follows.
    \qed
\end{inlineproof}

\begin{remark}
    Observe that for any $\epsilon > 0$ by taking $n > \log(1+ 2/\epsilon) $ we 
    have the competitive ratio of Algorithm~\ref{alg:approx1} is at least 
    $2 - \epsilon$.
\end{remark}

\section{Online Lower Bounds  for Two Robots}
\label{sec:lower_bounds_two_robots}

In this section we prove two lower bounds on the competitive ratio for arbitrary 
online algorithms. Our lower bounds require only two mobile robots. 
In the first lower bound (Theorem~\ref{thm:lowerbound1}) we assume that the 
speed of one of the robots is unknown and in the second 
(Theorem~\ref{thm:lowerbound2}) we assume that the starting position of one of 
the robots is unknown.
We provide both bounds (even though the second is slightly better) as the 
arguments are somewhat different and it seems plausible that an improved lower 
bound may come from  combining the two approaches. 

\begin{theoremrep}\label{thm:lowerbound1}
    The lower bound for the competitive ratio when the fast robot does not know 
    whether the speed of the slow robot is one or zero is at least $1.0391$.
\end{theoremrep}
\begin{appendixproof}
    In the proof we consider two robots $r_1$ and $r_v$ where $r_1$ is placed at 
    the source of the message and has speed either 0 or 1 as set by the 
    adversary and $r_v$ has speed greater than 1.
    Without loss of generality assume that $r_1$ and $r_v$ are at distance one. 
    By Lemma~\ref{lm:apollonius} the Apollonius circle is at distance 
    $\frac{v^2}{v^2 - 1}$ and has radius $\frac{v}{v^2 - 1}$. Let $C$ be the 
    center of the Apollonius circle. We place $C$ at the origin of the standard 
    orthogonal Cartesian coordinate system and $K$ (the position of $r_v$) at 
    position $(\frac{v^2}{v^2 - 1}, 0)$ as shown in 
    Figure~\ref{fig:lowerboundspeed}). 
    Let $D$ be the destination point and $M$ the meeting point when the speed of 
    $r_1$ is 1. 
    The ray $CM$ bisects the line segments $DM$ and $MK$.
    Let $M = \frac{v}{v^2 - 1} \cos\alpha$ and 
    $D = \frac{v^2}{v^2 - 1} \cos(2\alpha)$.

    Observe that if $r_v$ does not know the speed of $r_1$, then they can choose 
    a point $X$ to meet. If they meet, then $r_v$ takes the message to $D$. 
    However, if they don't meet at point $X$, then $r_v$ concludes that $r_1$ 
    has speed 0. Therefore, it reaches $S$ and then it moves toward $D$. 
    We claim that $X$ is in the circumference of the Apollonius circle. 
    Observe that any strategy where the fast robot waits is suboptimal since
    they can always meet at a point closer to the Apollonius circle. 
    Therefore, $X$ cannot be outside the Apollonius circle.
    Suppose by contradiction that $X'$ is inside the Apollonius circle. 
    Let $Y$ be the point in the line $KX'$ and the Apollonius circle. 
    From the triangle inequality, $|YD| \leq |YX'| + |X'D|$ and 
    $|YS| \leq |YX'| + |X'S|$ which contradicts the assumption.

    \begin{figure}[!htb]
        \begin{center}
        \includegraphics[scale=0.7]{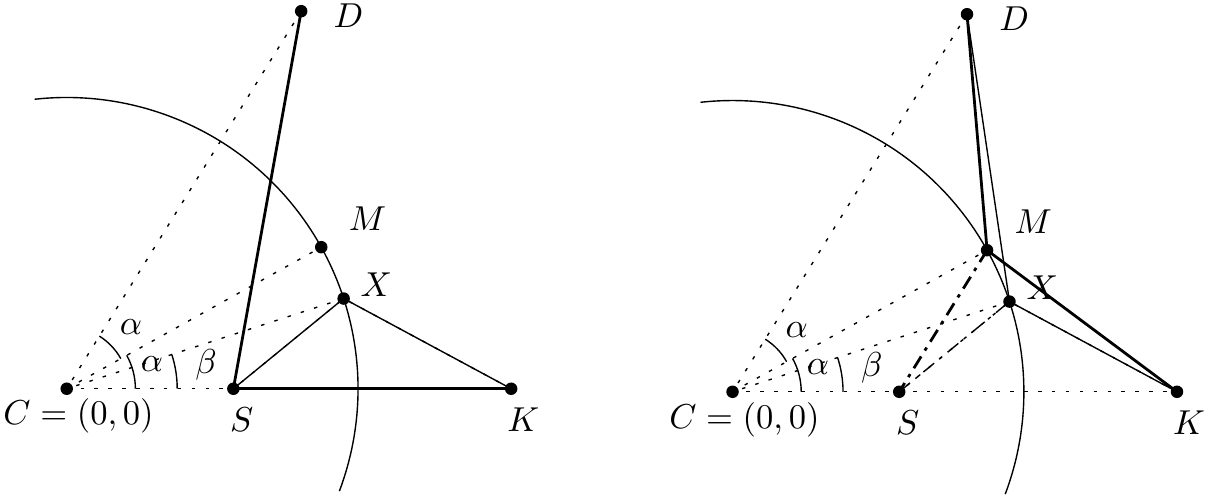}
        \end{center}
        \caption{%
            Configuration of two robots $r_1, r_v$ such that the speed of %
            $r_1$ is either $0$ or $1$ as determined by an adversary. %
            The left figure shows the trajectory of $r_v$ when $r_1$ has speed %
            0 and the right figure shows the trajectories of the robots when %
            $r_1$ has speed 1 (Bold  solid lines show the optimal trajectory %
            of $r_v$, bold dash-dotted lines show the optimal trajectory of %
            $r_1$, meanwhile solid lines show the trajectory of $r_v$ and dash %
            dotted line show the trajectory of $r_1$ of the online algorithm.%
        }
        \label{fig:lowerboundspeed}
    \end{figure}

    Let $\beta = \angle(XCK)$.
    Consider the case where the speed of $r_1$ is one. The competitive ratio 
    is: $$\frac{(|KX| + |XD|)/v}{\min(SD, (|KM|+|MD|)/v)}$$ .
    To calculate the distance we use  the Law of Cosines since 
    $|CK| = |CD| = \frac{v^2}{v^2-1}$, $|CM| = |CX|= \frac{v^2}{v^2-1}$ and 
    $|CS| = \frac{1}{v^2-1}$.
    Observe that  $r_1$ can deliver the message solo in optimal time when 
    $v |SD| \leq 1+ |SD|$ as well when $v|SD| < 2|KM|$.

    Now consider the case where the speed of $r_1$ is zero. 
    Here the competitive ratio is $$\frac{(|KX|+|XS|+|SD|)/v}{(|KS|+|SD|)/v)}$$.
    Since $X$ is at the circumference of the Apollonius circle, $|XS|=|KS|/v$.
    Observe that the maximum value occurs at the equilibrium point. 
    Simplifying we get, 
    $$
        \frac{|KX| + |XD|}{\min(v |SD|, 2|KM|)} 
            = \frac{|KX|( 1 + 1/v)+|SD|}{1+|SD|}
    $$.

    To obtain the best competitive ratio, for each speed $v \in (1,2]$, we 
    compute the angle $\alpha$ and $\beta$ that maximizes the competitive ratio. 
    The best competitive ratio is at least $1.0391$ that  occurs at speed 
    $v=1.65$ and angle $\alpha = 0.6597$ and $\beta = 0.2312$.
    \qed
\end{appendixproof}

\begin{theorem}\label{thm:lowerbound2}
    The lower bound for the competitive ratio when the slow robot does not know 
    the position of the fast robot is at least $1.04059$.
\end{theorem}    
 \begin{proof}
    Let the source $S$ be at the origin. Given $v$ and $\alpha$, we place the 
    destination $D$ at $$\left(\frac{1}{v^2-1} \sqrt{1+v^4-2 v^2 \cos(2\alpha)}, 0\right)$$
    as shown in Figure~\ref{fig:lowerboundposition}). 
    Let $C_1, C_2$ be the points at distance $\frac{1}{v^2 - 1}$ and 
    $\frac{v^2}{v^2 - 1}$  from $S$ and $D$ each. 
    Let $\beta = \angle(DSC_1) = \angle(DSC_2)$. 

    \begin{figure}[!htb]
        \begin{center}
        \includegraphics[scale=0.8]{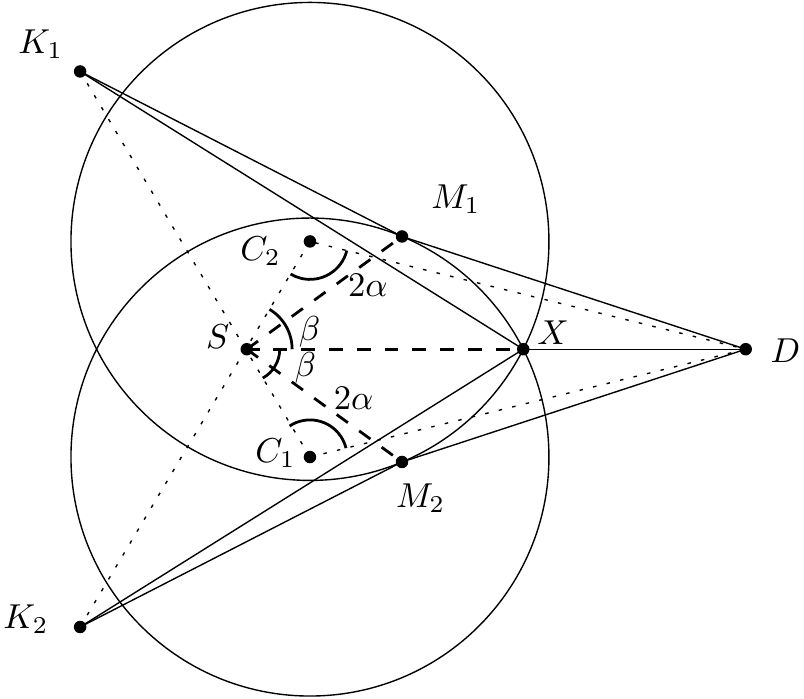}
        \end{center}
        \caption{%
            Configuration of two robots $r_1, r_v$ such that the position of %
            $r_v$ is unknown for $r_1$ and  determined by an adversary. 
            Solid lines show the trajectory of $r_v$, bold dashed lines show %
            the optimal trajectory of $r_1$ and dash line show the trajectory %
            of $r_1$.%
        }
        \label{fig:lowerboundposition}
    \end{figure}

    Let $K_1 = (-\cos\beta, \sin(\pi - \beta))$ and 
    $K_2 = -(\cos\beta, \sin(-\beta))$.
    Let $r_1$ be at $S$ and $r_v$ be at either $K_1$ or $K_2$ set by the 
    adversary.
    Observe that $C_1$ is the center of the Apollonius circle with radius 
    $\frac{v}{v^2 -1}$ if $r_v$ is at $k_1$ and $C_2$ is the  center of the 
    Apollonius circle with radius $\frac{v}{v^2 -1}$ if $r_v$ is at $k_2$.

    Let $M_1$  and $M_2$ be the meeting points of the optimal strategy. 
    Let $$|K_1M_1| = |K_2M_2| = \frac{v}{v^2 -1}\sqrt{ 1 + v^2 - 2v \cos\alpha}.$$
    Observe that if $r_1$ does not know the position of $r_v$, then it should 
    meet in a point that is at the same distance from the Apollonius circles. 
    Let $X$ be the meeting point of the online algorithm. 
    Observe that $X$ must lie on $SD$ otherwise the adversary will place the destination in the opposite side. 
    Arguing as in Lemma~\ref{lm:optimal_m} we can conclude that $X$ must be at the circumference of the Apollonius circle. 
    Then the competitive ratio is given by: 
    $$\frac{|SX|(v-1) +SD}{2 |K_1M_1|}$$.

    From the law of Sines $\sin\beta = |C_1D| \sin(2 \alpha) / |SD|$. 
    Consider the right triangles $\triangle(SX_1C_1)$ and $\triangle(XX_1C_1)$ 
    where $X_1$ is on the line $SD$. Let $|X_1C_1| = |SC_1| \sin\beta$. 
    Then, $|SX_1| = \sqrt{|SC_1|^2 -  |X_1C_1|^2}$ and 
    $X_1X = \sqrt{|C_1X|^2 - |X_1C_1|^2}$ where $|C_1X| = \frac{v}{v^2-1}$. 
    Observe that if $|C_1D|^2 \leq |SD|^2 + |C_1D|^2$, then $X_1$ is in between 
    $SX$ and, therefore, $|SX| =  |SX_1| + |X_1X|$, otherwise, 
    $|SX|=|X_1X|-|SX_1|$.

    To obtain the best competitive ratio, for each angle between $0$ and 
    $\pi/2$, we compute the best competitive ratio for all speeds. 
    The best competitive ratio is at least $1.04059$ that occurs at angle 
    $\alpha = 0.8953$ and $v = 2.7169$.

    \qed
\end{proof}

\section{Conclusion}
\label{sec:conclusion}

In this paper we studied the pony express communication problem for delivering 
a message from a source to a destination in the plane. We gave an optimal 
offline algorithm for the case of two robots and a $\sqrt{2}$ approximation for 
$n$ robots. We studied a particular simple online algorithm and provided tight 
bounds for its performance in both the two robot and $n$ robot case. 
Finally, we gave two distinct arguments for lower bounds on the competitive 
ratio of any online algorithm.

Our investigations leave a number of open problems. Of special interest is the 
complexity of the offline problem for the case of $n$ robots. 
While it seems unlikely, we can not even be sure the question of deciding if an 
instance can be solved in a given time bound is decidable as the equations we 
derive involve trigonometric functions. 
In the online setting, there remains a gap between the upper bounds and lower 
bounds in both the case of just two robots and the case of $n$ robots. 
Based upon this preliminary investigation, it would appear the exact bound will 
vary with the number of robots considered. 

Additional questions arise when one attempts to solve well-known communication 
tasks such as broadcast and converge-cast.
The delivery task considered in this paper made use of our ability to freely 
replicate the message which is to be delivered from the source to the 
destination. 
However, the problem is also interesting if the message to be delivered can be 
replicated only a given fixed number of times or even if it cannot be 
replicated at all (e.g., it is a physical object) in which case the problem 
resembles a transportation problem. Another interesting question would be to 
study the pony express communication problem in a setting where the robots may 
be subject to faults.  

\bibliographystyle{plain}
\bibliography{refs}

\end{document}